\documentclass[12pt]{article}
\usepackage{epsfig,amsmath,amssymb,amsfonts,amstext,amsthm,mathrsfs}
\usepackage{latexsym,graphics,epsf,epsfig,psfrag}
\usepackage{cite}
\usepackage{color}

\newtheorem{definition}{\textbf{Definition}}

\newtheorem{lemma}{\textbf{Lemma}}
\newtheorem{theorem}{\textbf{Theorem}}

\newcommand{\nn}{\nonumber}

\newcommand{\mE}{\mathrm{E}}
\newcommand{\Var}{\mathsf{Var}}

\newcommand{\cX}{\mathcal{X}}

\newcommand{\cY}{\mathcal{Y}}

\newcommand{\cZ}{\mathcal{Z}}

\DeclareMathAlphabet{\matheuf}{U}{euf}{m}{n}

\begin{document}

\vspace*{-2cm}

\begin{center}
  \baselineskip 1.3ex {\Large \bf The Capacity Region of the Source-Type Model for Secret Key and Private Key Generation
\footnote{The work of H. Zhang and Y. Liang was supported by a National Science Foundation
CAREER Award under Grant CCF-10-26565 and by the National Science Foundation under Grant CNS-11-16932. The work of L. Lai was supported by a National Science Foundation CAREER Award under Grant CCF-13-18980 and by the National Science Foundation under Grant CNS-13-21223.}\\
}
 \vspace{0.15in} Huishuai Zhang, Lifeng Lai, Yingbin Liang, Hua Wang
\footnote{H. Zhang and Y. Liang are with the Department of Electrical
Engineering and Computer Science, Syracuse University, Syracuse, NY 13244 USA (email: \{hzhan23,yliang06\}@syr.edu). L. Lai is with the Department of Electrical and Computer Engineering, Worcester Polytechnic Institute, Worcester, MA 01609 USA (email: llai@wpi.edu). H. Wang is with Qualcomm Inc., Bridgewater, NJ 08807 USA (email: huaw@qti.qualcomm.com) }
\end{center}

\begin{abstract}
The problem of simultaneously generating a secret key (SK) and private key (PK) pair among three terminals via public discussion is investigated. In this problem, each terminal observes a component of correlated sources. All three terminals are required to generate the common SK to be concealed from an eavesdropper that has access to the public discussion, while two designated terminals are required to generate an extra PK to be concealed from both the eavesdropper and the remaining terminal. An outer bound on the SK-PK capacity region was established by Ye and Narayan in \cite{Ye05con}, and was shown to be achievable for a special case. In this paper, the SK-PK capacity region is established in general by developing schemes to achieve the outer bound for the remaining two cases. The main technique lies in the novel design of a random binning-joint decoding scheme that achieves the existing outer bound.
\end{abstract}


\section{Introduction}\label{sec:introduction}
The problem of secret key generation via public discussion under the source model was initiated by~\cite{Csiszar93,Maurer93}, which established a remarkable fact that two terminals, each possessing correlated but not exactly the same observations, can establish a shared secret key by only talking to each other in the public. In the basic source-type model considered in \cite{Csiszar93,Maurer93}, there are two legitimate terminals, who observe correlated source sequences and can communicate with each other through a public channel, and eavesdroppers, who have perfect access to the public channel. The main observation is that, because of the correlation, terminal $\mathcal{X}$ can recover terminal $\mathcal{Y}$'s source sequence by letting terminal $\mathcal{Y}$ send limited amount of information using distributed source coding technique~\cite{Slepian73}. Then both terminal $\mathcal{X}$ and terminal $\mathcal{Y}$ can generate a shared secret key based on terminal $\mathcal{Y}$'s source sequence subtracting the information that has been revealed. The close connection between the distributed source coding and secret key generation also holds on more general source-type models~\cite{Csiszar04}. In particular,~\cite{Csiszar04} studied a general network with multiple terminals, in which a subset of terminals need to generate a shared secret key. \cite{Csiszar04} showed that the secret key capacity is equal to the joint entropy of all source observations subtracting the minimum amount of information needed to enable the subset of terminals to recover all source observations.

Until now, with few exceptions to be discussed in the sequel, most of the existing studies focused on generation of a single key~\cite{Maurer99, Csiszar00, Maurer03a, Maurer03b, Maurer03c, Csiszar04}. However, there are various practical scenarios in which multiple keys need to be simultaneously generated. For instance, a number of terminals can have different security clearance levels, and each terminal is allowed to access confidential documents up to its own clearance level. Terminals with the same clearance level should share the same key, and should be kept ignorant of higher level keys.

There have been several existing studies that addressed generation of multiple keys~\cite{Ye05con,Lai13,Lai:ITW:12, Ye05dis}. Being of particular interest to us, Ye and Narayan studied a multi-key source-type model in~\cite{Ye05con}, in which three terminals (say terminals $\mathcal{X}, \mathcal{Y}$ and $\mathcal{Z}$) observe correlated source sequences, and wish to generate a common secret key (SK) among all of them, which should be concealed from eavesdroppers, and a private key (PK) between $\mathcal{X}$ and $\mathcal{Y}$ that should be concealed from $\mathcal{Z}$ and eavesdroppers.~\cite{Ye05con} provided both outer and inner bounds on the SK-PK capacity region. In particular, the outer bound has three different forms corresponding respectively to three cases of correlations among the sources. In \cite{Ye05con}, it was shown that the outer bound is achievable for one case, and hence the SK-PK capacity region was established for this case. However, for the other two cases, there are gaps between the outer bound and the inner bound derived based on the scheme developed in~\cite{Ye05con}. Finding schemes to achieve the outer bound for the other two cases was left as an open problem in \cite{Ye05con}. In fact, the outer bound in the other two cases suggests the necessity of a scheme such that $\mathcal{X}$ helps $\mathcal{Y}$ to recover $\mathcal{Z}$'s information without revealing any more information of $\mathcal{Z}$ to public. This is the major technical challenge to obtain the SK-PK capacity region in general.

Our main contribution in this paper lies in finding schemes that achieve the outer bound for the other two cases for the SK-PK source-type model in \cite{Ye05con}. Then, combined with the result in \cite{Ye05con} for the first case, the full SK-PK capacity region is established. In order to address the technical challenge mentioned above, we design schemes such that terminal $\mathcal{X}$ helps to improve the quality of the side information at $\mathcal{Y}$ in recovering $\mathcal{Z}$'s information rather than directly revealing information of $\mathcal{Z}.$



The paper is organized as follows. Section \ref{sec:model} contains the model description. Section \ref{sec:mainresults} presents our main results on the SK-PK capacity region. Section \ref{sec:proof2} and \ref{sec:proof3} consist of the proofs of the main theorem for the two cases, respectively. Section \ref{sec:conclusion} provides some concluding remarks.

\section{System Model}\label{sec:model}

Consider a discrete memoryless source, whose outputs at each time instant are generated based on the joint distribution of random variables $(X, Y, Z)$ with corresponding alphabets $(\mathcal{X, Y, Z})$. We consider a system with three terminals $(\mathcal{X, Y, Z})$ and an eavesdropper. Here, we use the alphabet symbols to denote the terminals. Terminal $\mathcal{X}$ observes $n$ independent and identically distributed~(i.i.d.) repetitions of $X$, i.e., $X^n=(X_1,\cdots, X_n)$, and terminals $\mathcal{Y}$ and $\mathcal{Z}$ observe $Y^n=(Y_1,\cdots, Y_n)$ and $Z^n=(Z_1,\cdots, Z_n)$, respectively. We assume that the eavesdropper does not have source observations and terminals are allowed to communicate with each other over a public noiseless channel with no rate constraint. We further assume that all transmissions over the public channel are observable to all parties including the eavesdropper. The public discussion can be interactive. Without loss of generality, we assume that terminals $(\mathcal{X, Y, Z})$ take turns to transmit for $r$ rounds over $3r$ consecutive time slots. We use $3r$ random variables $F_1,\cdots,F_{3r}$ to denote these transmissions, where $F_t$ denotes the transmission in time slot $t$ for $1\le t \le 3r$. The transmission $F_t$ can be any function of its own observation and all previous transmissions $F_{[1,t-1]}=(F_1,\cdots, F_{t-1})$. We use $\mathbf{F}=(F_1,\cdots, F_{3r})$ to denote all transmissions in $3r$ time slots. Furthermore, we note that although our result in this paper also holds for the case that allows additional randomization at each terminal, we do not explicitly allow such randomization in our model for simplicity.

In this system (see Fig.~\ref{fig:systemmodel}), terminals $\mathcal{X, Y}$ and $\mathcal{Z}$ wish to generate a common secret key $K_S$, which is required to be kept secure from the eavesdropper (that has access to only the public discussion). Furthermore, terminals $\mathcal{X}$ and $\mathcal{Y}$ wish to generate a private key $K_P$, which is required to be kept secure not only from the eavesdropper but also from terminal $\mathcal{Z}$. 
\begin{figure}[thb]
	\centering
	\includegraphics[width=3.7in]{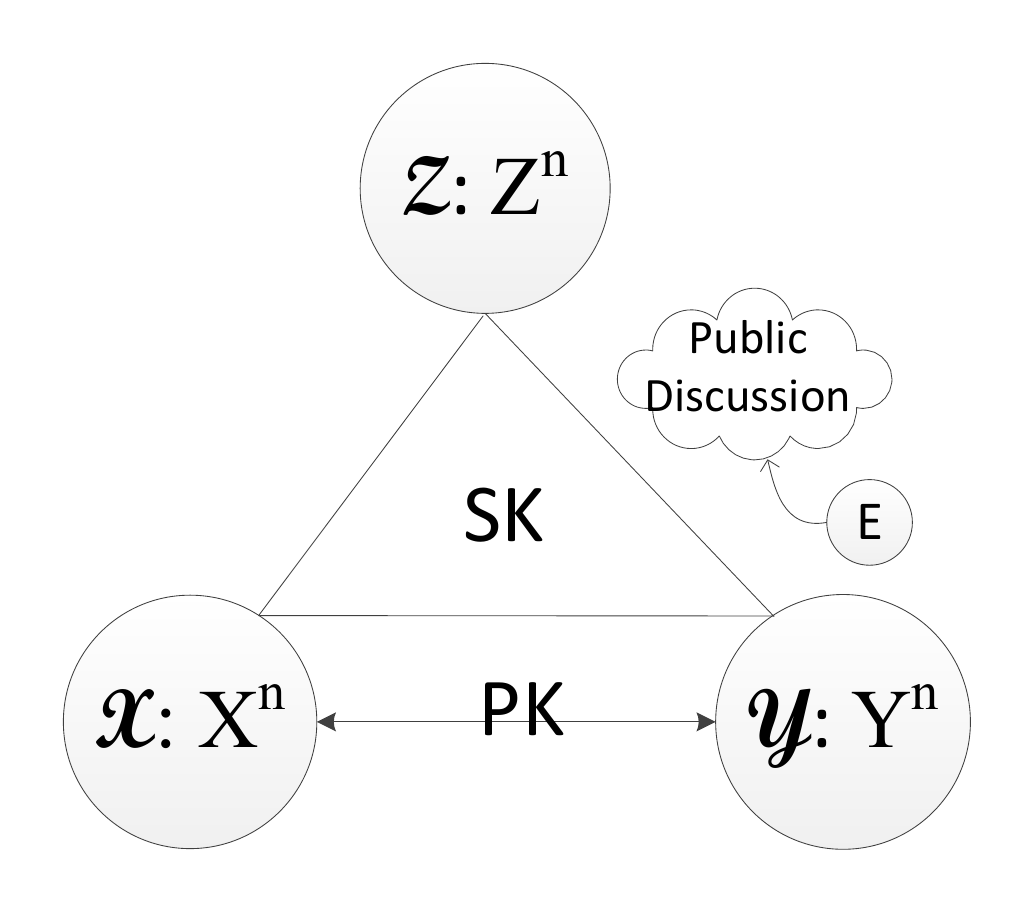}\\
	\caption{System model}
    	\label{fig:systemmodel}
\end{figure}

We next introduce the mathematical definition of the secret key and the private key. A random variable $U$ is said to be $\epsilon$-{\em recoverable} from another random variable $V$, if there exists a function $f$ such that
\begin{equation} \label{eq:eps}
	\Pr\{U\neq f(V)\}<\epsilon.
\end{equation}

\begin{definition}
A pair $(K_S, K_P)$ is said to be an $\epsilon$-(SK, PK) if $K_S$ and $K_P$ satisfy the following requirements.

$\bullet$ $K_S$ is $\epsilon$-{\em recoverable} at each of the three terminals with the public transmission $\mathbf{F}$, i.e., it can be $\epsilon$-{\em recoverable} from $(X^n, \mathbf{F})$, $(Y^n, \mathbf{F})$ and $(Z^n, \mathbf{F})$, respectively;

$\bullet$ $K_P$ is $\epsilon$-{\em recoverable} at terminals $\mathcal{X}$ and $\mathcal{Y}$ with public transmission $\mathbf{F}$, i.e., it can be $\epsilon$-{\em recoverable} from $(X^n, \mathbf{F})$ and $(Y^n, \mathbf{F})$, respectively;

$\bullet$ $K_S$ and $K_P$ satisfy the secrecy condition
\begin{flalign}
	&\frac{1}{n}I(K_S;\mathbf{F})<\epsilon, \label{eq:SKsecrecy}\\
	&\frac{1}{n}I(K_P;\mathbf{F},Z^n)<\epsilon \label{eq:PKsecrecy}
\end{flalign}
for large enough $n$, where $\epsilon$ can be arbitrarily small; and

$\bullet$ $K_S$ and $K_P$ satisfy the uniformity condition
\begin{flalign}
	\frac{1}{n}H(K_S)\ge \frac{1}{n}\log |\mathcal{K_S}|-\epsilon, \label{eq:SKuniform} \\
	\frac{1}{n}H(K_P)\ge \frac{1}{n}\log |\mathcal{K_P}|-\epsilon,
\end{flalign}
for large enough $n$, where $|\mathcal{K_S}|$ and $|\mathcal{K_P}|$ denote the alphabet sizes of the random variable $K_S$ and $K_P$, respectively.
\end{definition}
%

We note that the secrecy conditions \eqref{eq:SKsecrecy} and \eqref{eq:PKsecrecy} are in the weak sense, and can be strengthened to the strong sense without loss of performance as in \cite{Maurer00}.

\begin{definition}
A rate pair $(R_S, R_P)$ is said to be an achievable SK-PK rate pair if for every $\epsilon>0$, $\delta>0$, and for sufficiently large $n$, there exists an $\epsilon$-(SK,PK) pair $(K_S^{(n)},K_P^{(n)})$ such that
 \begin{flalign}
	\frac{1}{n} H(K_S^{(n)})>R_S-\delta, \ \ \ \ \ \ \ \ \ \ \ \frac{1}{n} H(K_P^{(n)})>R_P-\delta.
\end{flalign}
\end{definition}
Our goal is to characterize the {\em SK-PK capacity region} that contains all achievable rate pairs $(R_S,R_P)$.

\section{Main Results}\label{sec:mainresults}

\subsection{Preliminaries}

The model introduced in Section \ref{sec:model} has been studied by Ye and Narayan in \cite{Ye05con}, which provided outer and inner bounds on the SK-PK capacity region (see Chapter 3 in \cite{Ye05dis} for more details). We cite the outer bound in \cite{Ye05con} below, which is useful for presenting our results in the next subsection. For notational convenience, we define
\begin{flalign}
	R_A&:=I(Z;XY), \label{eq:a} \\
	R_B&:=\min\{I(X;YZ), I(Y;XZ)\},  \\
	R_C&:=\frac{1}{2}(H(X)+H(Y)+H(Z)-H(X,Y,Z)). \label{eq:c}
\end{flalign}
\begin{theorem}~\cite{Ye05con}\label{th:CapReg}
An outer bound on the SK-PK capacity region for the model in Section \ref{sec:model} contains the rate pairs $(R_S,R_P)$ satisfying
	\begin{flalign}
		&R_S\le R_A,\label{eq:outbound1}\\
		&R_P \le I(X;Y|Z),\label{eq:outbound2}\\
		&R_S+R_P \le R_B,\label{eq:outbound3}\\
		&2R_S+R_P \le 2R_C.\label{eq:outbound4}
	\end{flalign}
where the constants $R_A, R_B$ and $R_C$ are defined in \eqref{eq:a}-\eqref{eq:c}.
\end{theorem}
It is instructional to first note a few observations about the above outer bound.

1. If we dedicate to generate the private key $K_P$ without considering the secret key $K_S$, then the model becomes the private key model studied in \cite{Csiszar04}. The outer bound on $R_P$ is \eqref{eq:outbound2}, which can be achieved by letting terminal $\mathcal{Z}$ reveal all its information to public. Here terminal $\mathcal{Z}$ is curious but honest, and helps to generate the private key.

2. If we dedicate to generate the secret key $K_S$ without considering the private key $K_P$, then the model reduces to the secret key model studied in \cite{Csiszar04}. Correspondingly the above outer bound reduces to $R_S\le \min\{R_A,R_B,R_C\}$ based on \eqref{eq:outbound1}, \eqref{eq:outbound3} and \eqref{eq:outbound4}. According to \cite{Csiszar04}, this bound is achievable by applying the ``omniscience'' scheme, which requires each terminal recover the sources of all three terminals after the public discussion.

3. The sum rate bound \eqref{eq:outbound3} can be viewed as a cut-set type bound, because both $\mathcal{X}$ and $\mathcal{Y}$ need to generate two keys $K_S$ and $K_P$ simultaneously.


We next further explain the above outer bound in detail. We note that this outer bound can take three different structures corresponding respectively to the following three cases: {\bf case 1} with $R_B=\min\{R_A, R_B, R_C\}$, {\bf case 2} with $R_C=\min\{R_A, R_B, R_C\}$, and {\bf case 3} with $R_A=\min\{R_A, R_B, R_C\}$.

For case 1, it was shown in \cite{Ye05con} that the outer bound (as illustrated in Fig.~\ref{fig:case1}) is achievable. It is clear that the point B with the rate coordinates $(R_B,0)$ is achievable by applying the ``omniscience" scheme in\cite{Csiszar04} and the point E with the rate coordinates $(0, I(X;Y|Z))$ is achievable by letting $\mathcal{Z}$ reveal all of its information to public. The corner point T with the rate coordinates $(R_B-I(X;Y|Z), I(X;Y|Z))$ is shown to be achievable in \cite{Ye05con}. The idea is to let $\mathcal{Z}$ reveal information at rate $R_\mathcal{Z}=\max\{H(Z|X), H(Z|Y)\}$ so that both $\mathcal{X}$ and $\mathcal{Y}$ can recover $Z^n$ correctly with probability close to 1. Now $Z^n$ is the information shared by three terminals, and hence the secret key $K_S$ can be generated based on $Z^n$ with rate $R_S=H(Z)-R_\mathcal{Z}=\min\{I(X;Z), I(Y;Z)\}$. Then, given $Z^n$, terminals $\mathcal{X}$ and $\mathcal{Y}$ can generate a private key with rate $R_P=I(X;Y|Z)$ if terminal $\mathcal{X}$ reveals information at rate $R_\mathcal{X}=H(X|YZ)$ to terminal $\mathcal{Y}$. Finally, the entire outer bound can be achieved by time-sharing scheme.
\begin{figure}[thb]
	\centering
	\includegraphics[width=3.7in]{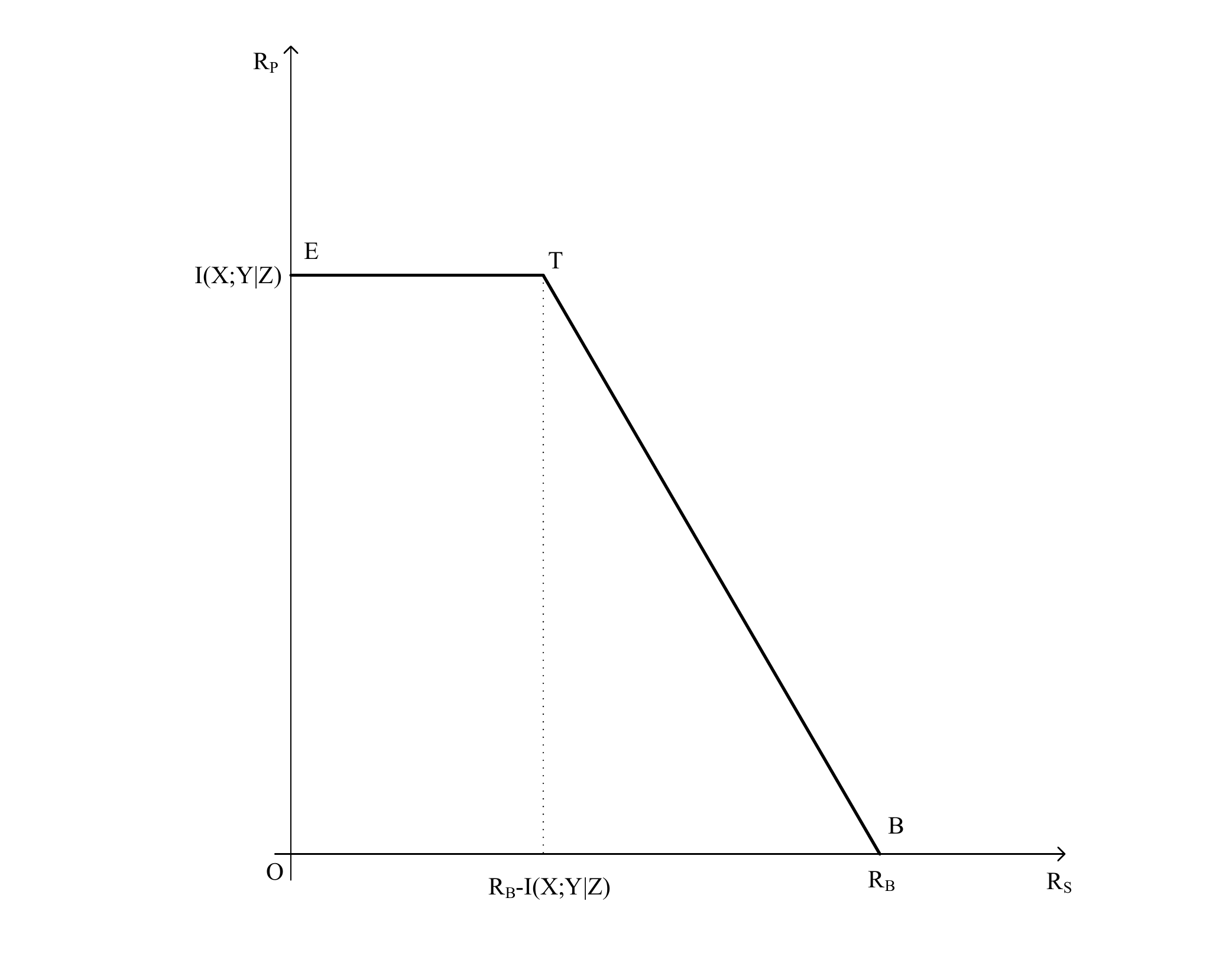}\\
	\caption{Out bound for case 1: the quadrangle O-E-T-B-O}
    	\label{fig:case1}
\end{figure}

In this paper, we show that the outer bound can be achieved for cases 2 and 3. Thus, this outer bound is the SK-PK capacity region in general.

\subsection{Main Theorem}\label{ssec:mainth}

Our main contribution in this paper lies in finding schemes that achieve the outer bound in Theorem \ref{th:CapReg} for cases 2 and 3. Thus, combined with the result in \cite{Ye05con} for case 1, the SK-PK capacity region is established in general. We provide our main result in the following theorem.
\begin{theorem}\label{th:TightTh}
The outer bound in Theorem \ref{th:CapReg} is achievable for cases 2 and 3, and hence is the SK-PK capacity region for the model given in Section \ref{sec:model} in general.
\end{theorem}
We next provide general ideas for the design of achievable schemes for cases 2 and 3. The detailed proof is provided in Sections \ref{sec:proof2} and \ref{sec:proof3}.

In case 2, $R_C=\min\{R_A,R_B,R_C\}$. The outer bound in Theorem \ref{th:CapReg} is plotted in Fig.~\ref{fig:case2} as the pentagon O-E-T-P-C-O. It has been shown in \cite{Ye05con} that the corner points E, T and C are achievable. It is thus sufficient to show that the point P is achievable. Then the entire pentagon can be achieved by time sharing.

\begin{figure}[thb]
	\centering
	\includegraphics[width=3.7in]{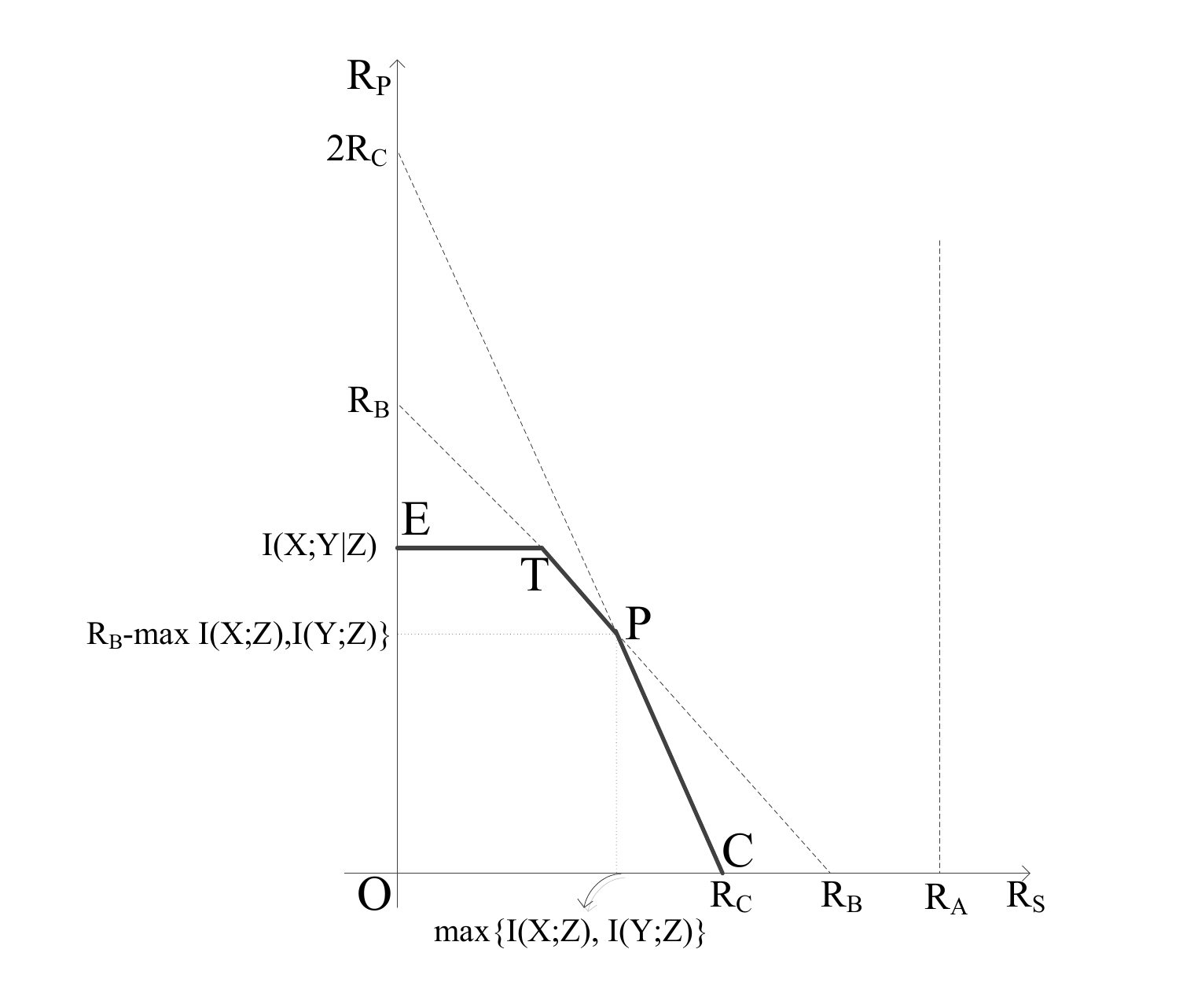}\\
	\caption{Outer bound for case 2: the pentagon O-E-T-P-C-O}
    	\label{fig:case2}
\end{figure}

\begin{figure}[thb]
	\centering
	\includegraphics[width=3.7in]{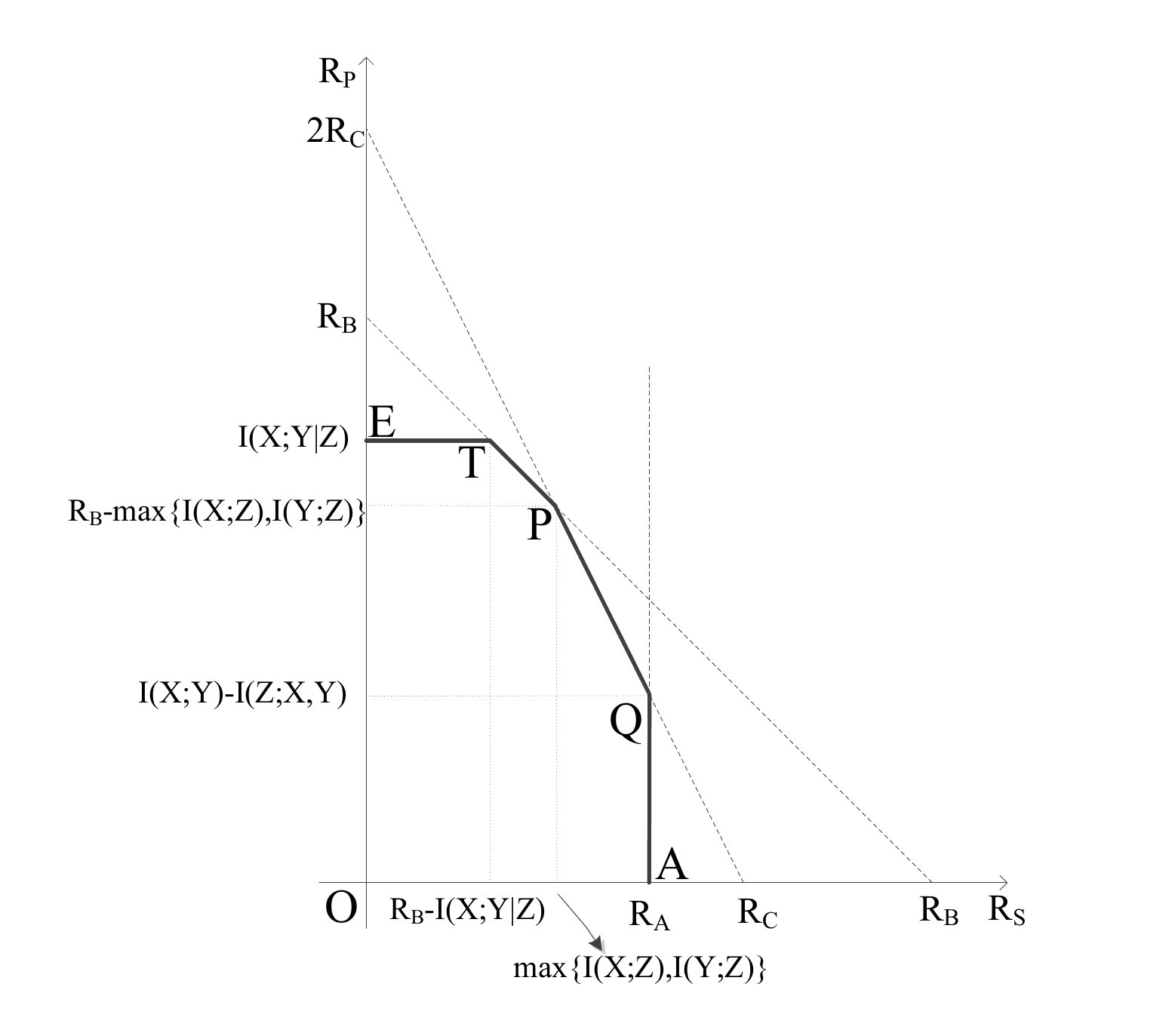}\\
	\caption{Outer bound for case 3: the hexagon O-E-T-P-Q-A-O}
    	\label{fig:case3}
\end{figure}

We note that the rate coordinates of the point P is $$\Big(\max\{I(X;Z), I(Y;Z)\}, R_B-\max\{I(X;Z), I(Y;Z)\}\Big).$$  Without loss of generality, we assume that $I(X;Z)>I(Y;Z)$ (the argument for the opposite assumption is similar), and hence $R_B=I(Y;XZ)$ and the point P becomes $(I(X;Z), I(Y;XZ)-I(X;Z))$. The SK rate $R_S=I(X;Z)$ suggests that the highest rate that $\mathcal{Z}$ can transmit publicly is $H(Z|X)$, with which $\mathcal{X}$ recovers $Z^n$, but $\mathcal{Y}$ cannot recover $Z^n$. Then $\mathcal{X}$ must transmit some information to help $\cY$ to recover $Z^n$ so that all three terminals can generate a secret key based on $Z^n$. Furthermore, the information transmitted by terminal $\cX$ also helps $\cY$ to recover $X^n$ so that $\cX$ and $\cY$ can generate a private key. The critical part of our achievable scheme lies in that terminal $\mathcal{X}$'s transmission should help $\mathcal{Y}$ to recover $Z^n$ without revealing more information about $Z^n$ to public beyond terminal $\cZ$'s transmission. Otherwise, the SK rate $R_S=I(X;Z)$ is not achievable. The idea is that $\cX$ helps $\cY$ to improve its resolvability of $Z^n$ rather than revealing information about $Z^n$ directly. Section \ref{ssec:understanding2} provides further technical intuition of the achievable scheme based on typicality arguments.



In case 3, $R_A=\min\{R_A,R_B,R_C\}$. The outer bound in Theorem \ref{th:CapReg} is plotted in Fig.~\ref{fig:case3} as the hexagon O-E-T-P-Q-A-O. It has been shown in \cite{Ye05con} that the corner points E, T and A are achievable. The point P can be achieved by applying the same scheme as in case 2. It is thus sufficient to show that the point Q is achievable. Then the entire hexagon can be achieved by time sharing.


The rate coordinates of the point Q is given by $(I(Z;XY), I(X;Y)-I(Z;XY))$. The SK rate $I(Z;XY)$ suggests that the highest rate that $\mathcal{Z}$ can transmit publicly is $H(Z|XY)$, with which neither $\mathcal{X}$ nor $\mathcal{Y}$ can recover $Z^n$. Then both $\mathcal{X}$ and $\mathcal{Y}$ must help each other to recover $Z^n$ so that all three terminals can generate a secret key based on $Z^n$. Furthermore, terminal $\cY$ also helps terminal $\cX$ to recover $Y^n$ so that $\cX$ and $\cY$ can generate a private key. The critical part lies in that $\mathcal{X}$ and $\mathcal{Y}$'s transmission help each other to recover $Z^n$ without revealing more information about $Z^n$ to public beyond terminal $\cZ$'s transmission. Otherwise, the SK rate $R_S=I(Z;XY)$ is not achievable. The idea is that $\cX$ and $\cY$ help each other to improve their resolvability of $Z^n$ rather than revealing information about $Z^n$ directly. Section \ref{ssec:understanding3} provides further technical intuition of the achievable scheme based on typicality arguments.


%
\section{Achievability Proof for Case 2}\label{sec:proof2}
In this section, we provide the achievability proof for case 2 with subsection \ref{ssec:technique2} containing the technical proof and subsection \ref{ssec:understanding2} containing further intuitive justification.

\subsection{Technical Proof}\label{ssec:technique2}
In this subsection, we show that the outer bound given in Theorem \ref{th:CapReg} for case 2 is achievable.
In this case, $R_C=\min\{R_A, R_B, R_C\}$. We assume that $R_C<R_B$, which implies
\begin{flalign}
	&I(X;Y)\le I(Z;XY),\\
	&I(X;Z)<I(Y;XZ),\label{eq:case2ass2}\\
	&I(Y;Z)<I(X;YZ).
\end{flalign}
The case of equality with $R_C=R_B$ reduces to case 1.

The outer bound for case 2 is plotted in Fig.~\ref{fig:case2} as the pentagon O-E-T-P-C-O. As we mentioned in Section \ref{ssec:mainth} it has been shown in \cite{Ye05con} that the corner points E, T and C are achievable. It is thus sufficient to show that the point P is achievable. Then the entire pentagon can be achieved by time sharing. We note that the rate coordinate corresponding to the point P is $(\max\{I(X;Z), I(Y;Z)\}, R_B-\max\{I(X;Z), I(Y;Z)\})$. Without loss of generality, we assume that $I(X;Z)>I(Y;Z)$ (the argument when $I(X;Z)<I(Y;Z)$ is similar), and hence $R_B=I(Y;XZ)$ and the point P becomes $(I(X;Z), I(Y;XZ)-I(X;Z))$. Our scheme to achieve point P is based on random binning and joint typicality.

\underline{Codebook Generation}: At terminal $\mathcal{Z}$, randomly and independently assign a bin index $f$ to each sequence $z^n\in \mathcal{Z}^n$, where $f\in [1:2^{nR_\mathcal{Z}}]$ with $R_\mathcal{Z}$ given by
\begin{equation}\label{eq:case2Rz}
R_\mathcal{Z}=H(Z|X)+\epsilon.
\end{equation}
We use $f(z^n)$ to denote the bin index of the sequence $z^n$, and use $B_\mathcal{Z}(f)$ to denote the bin indexed by $f$. Then randomly and independently assign a sub-bin index $\phi$ to each sequence in each nonempty bin $B_\mathcal{Z}(f)$, where $\phi \in [1:2^{nR_S}]$
with $R_S$ given by
\begin{equation}\label{eq:case2Rs}
R_S=I(X;Z)-2\delta(\epsilon)-2\epsilon.
\end{equation}
We further use $B_\mathcal{Z}(f,\phi)$ to denote the sub-bin indexed by $\phi$ within the bin $B_\mathcal{Z}(f)$.

At terminal $\mathcal{X}$, randomly and independently assign a bin index $g$ to each sequence $x^n\in \mathcal{X}^n$, where $g\in[1:2^{nR_\mathcal{X}}]$ with $R_\mathcal{X}$ given by
\begin{equation}\label{eq:case2Rx}
R_\mathcal{X}=H(XZ|Y)-H(Z|X).
\end{equation}
We use $g(x^n)$ to denote the bin index of the sequence $x^n$, and use $B_\mathcal{X}(g)$ to denote the bin indexed by $g$. Then randomly and independently assign a sub-bin index $\psi$ to each sequence in each nonempty bin $B_\mathcal{X}(g)$, where $\psi \in[1: 2^{nR_P}]$
with $R_P$ given by
\begin{equation}\label{eq:case2Rp}
R_P=I(XZ;Y)-I(X;Z)-2\delta(\epsilon)-\epsilon.
\end{equation}
We further use $B_\mathcal{X}(g,\psi)$ to denote the sub-bin indexed by $\psi$ within the bin $B_\mathcal{X}(g)$.

It can be verified that $R_\mathcal{X}<H(X|Z)$ and $R_P>0$ based on the case assumption \eqref{eq:case2ass2}.

This codebook assignment is known by all parties, i.e., terminals $\mathcal{X},\mathcal{Y},\mathcal{Z}$ and the eavesdropper.

\underline{Encoding and Transmission}:
Given a sequence $z^n$, terminal $\mathcal{Z}$ finds the index pair $(f,\phi)$ such that $z^n\in B_\mathcal{Z}(f,\phi)$, and then reveals the index $f=f(z^n)$ over the public channel to all parties, i.e., terminals $\mathcal{X},\mathcal{Y}$ and the eavesdropper.

Given a sequence $x^n$, terminal $\mathcal{X}$ finds the index pair $(g,\psi)$ such that $x^n\in B_\mathcal{X}(g,\psi)$, and then reveals the index $g=g(x^n)$ over the public channel to all parties, i.e., terminals $\mathcal{Y},\mathcal{Z}$ and the eavesdropper.

\underline{Decoding}:
The decoding scheme is based on the joint typicality. We use $T^{(n)}_\epsilon(P_{XYZ})$ to denote the strongly joint $\epsilon$-typical set based on the joint distribution $P_{XYZ}$.

Terminal $\mathcal{X}$, given $x^n$ and the bin index $f$, claims $\tilde{z}^n$ as recovery of $z^n$ if there exists a unique sequence $\tilde{z}^n\in B_\mathcal{Z}(f)$ that satisfies $(\tilde{z}^n,x^n) \in T^{(n)}_\epsilon(P_{XZ})$, or claims decoding failure otherwise.


Terminal $\mathcal{Y}$, given $y^n$ and the bin indexes $f$ and $g$, claims $(\hat{z}^n, \hat{x}^n)$ as recovery of $(z^n, x^n)$, if there exist a unique pair of sequences $(\hat{z}^n,\hat{x}^n)$ such that $\hat{z}^n\in B_\mathcal{Z}(f)$, $\hat{x}^n\in B_\mathcal{X}(g)$, and $(\hat{x}^n,\hat{z}^n,y^n)\in T^{(n)}_\epsilon(P_{XYZ})$, or claims decoding failure otherwise.

Due to \eqref{eq:case2Rz} and \eqref{eq:case2Rx}, it can be verified that $R_\mathcal{Z}>H(Z|XY)$, $R_\mathcal{X}>H(X|YZ)$ and $R_\mathcal{X}+R_\mathcal{Z}>H(XZ|Y)$ which implies the following inequalities hold according to the result of distributed source coding problem in \cite{Slepian73,Csiszar04,Gamal11}:
\begin{equation}\label{eq:errorx}
	\Pr\{Z^n\neq \tilde{Z}^n\}<\epsilon,
\end{equation}
\begin{equation}\label{eq:errory}
	\Pr\{X^n \neq \hat{X}^n\ or\ Z^n \neq \hat{Z}^n\}<\epsilon.
\end{equation}

\underline{Key Generation}: Terminal $\mathcal{Z}$ claims $K_S=\phi(Z^n)$. Terminal $\mathcal{X}$ claims $\tilde{K}_S=\phi(\tilde{Z}^n)$ and $K_P=\psi(X^n)$. Terminal $\mathcal{Y}$ claims $\hat{K}_S=\phi(\hat{Z}^n)$ and $\hat{K}_P=\psi(\hat{X}^n)$. Due to \eqref{eq:errorx} and \eqref{eq:errory}, we have
\begin{flalign}
	&\Pr\{K_S=\tilde{K}_S=\hat{K}_S\}>1-\epsilon,\\
 	&\Pr\{K_P=\hat{K}_P\}>1-\epsilon.
\end{flalign}

\underline{Analysis of Secrecy}:
We evaluate the leakage key rate averaged over the random codebook ensemble. Due to \eqref{eq:errorx} and \eqref{eq:errory}, in order to prove that the secrecy requirements \eqref{eq:SKsecrecy} and \eqref{eq:PKsecrecy} hold, it is sufficient to show the following two inequalities hold:
\begin{flalign}
 	&\frac{1}{n}I(K_S; \mathbf{F}|\mathcal{C})<\epsilon,\\
 	&\frac{1}{n}I(K_P; \mathbf{F}Z^n|\mathcal{C})<\epsilon.
\end{flalign}
 To simplify notations, let $f:=f(Z^n)$, and $g:=g(X^n)$. Hence $f$ and $g$ are random variables transmitted over the public channel, where the randomness is not only due to random realizations of the source sequences, but also due to random binning assignments~(i.e., random codebook generation). It is also clear that the public transmission $\mathbf{F}=\{f,g\}$. We further let $\phi:=\phi(Z^n)$ and $\psi:=\psi(X^n)$. Hence, $K_P=\psi$ and $K_S=\phi$. Then, we have
 \begin{flalign}
 	I(K_S;\mathbf{F}|\mathcal{C})&=I(\phi;f,g|\mathcal{C}) \nn \\
 	&=I(\phi;f|\mathcal{C})+I(\phi;g|f,\mathcal{C}) \nn \\
 	&\le I(\phi;f|\mathcal{C})+I(\phi,f;g|\mathcal{C}) \nn \\
 	&\stackrel{(a)}{\le} I(\phi;f|\mathcal{C})+I(Z^n;g|\mathcal{C}) \label{eq:case2ksbound}
\end{flalign}
where (a) follows from the fact that given $z^n$ and codebook $\mathcal{C}$, $f$ and $\phi$ are deterministic and independent from all other variables. We further derive
\begin{flalign}
 	I(K_P;\mathbf{F},Z^n|\mathcal{C})&=I(\psi;f,g,Z^n|\mathcal{C}) \nn \\
 	&=I(\psi;g,Z^n|\mathcal{C}) \nn \\
 	&=I(\psi;g|\mathcal{C})+I(\psi;Z^n|g,\mathcal{C}) \nn \\
 	&\le I(\psi;g|\mathcal{C})+I(\psi,g;Z^n|\mathcal{C}). \label{eq:case2kpbound}
\end{flalign}
We next show that each of the three terms $I(\phi;f|\mathcal{C}), I(\psi;g|\mathcal{C})$ and $I(\psi,g;Z^n|\mathcal{C})$ can be arbitrarily small for $n$ large enough. We first consider
\begin{equation*}
	\begin{split}
		I(\phi;f|\mathcal{C})&=I(\phi,Z^n;f|\mathcal{C})-I(Z^n;f|\phi,\mathcal{C})\\
		&=I(Z^n;f|\mathcal{C})-I(Z^n;f|\phi,\mathcal{C})\\
		&=H(Z^n|\mathcal{C})-H(Z^n|f,\mathcal{C})-H(Z^n|\phi,\mathcal{C})+H(Z^n|f,\phi,\mathcal{C}).
	\end{split}
\end{equation*}
It is clear that
\begin{flalign*}
	H(Z^n|f,\mathcal{C})&=H(Z^n,f|\mathcal{C})-H(f|\mathcal{C})=H(Z^n|\mathcal{C})-H(f|\mathcal{C})\ge H(Z^n|\mathcal{C})-nR_\mathcal{Z}
\end{flalign*}
Similarly, we have $H(Z^n|\phi,\mathcal{C})\ge H(Z^n|\mathcal{C})-nR_S$. Thus,
\begin{equation}
I(\phi;f|\mathcal{C})\le n(R_S+R_\mathcal{Z}-H(Z))+H(Z^n|f,\phi,\mathcal{C})
\end{equation}
where we used the fact that $Z^n$ is independent from $\mathcal{C}$, and hence $H(Z^n|\mathcal{C})=nH(Z)$.
In order to bound the last term, we introduce the following useful lemma.
\begin{lemma}\label{le:1}
	If $R_S+R_\mathcal{Z}<H(Z)-2\delta(\epsilon)$, then
	\begin{equation*}
 		\limsup_{n \rightarrow \infty}\frac{1}{n} H(Z^n|f,\phi,\mathcal{C})<H(Z)-R_S-R_\mathcal{Z}+\delta(\epsilon)
	\end{equation*}
\end{lemma}
\begin{proof} See Appendix \ref{app:a}. \end{proof}

Following from Lemma \ref{le:1} and \eqref{eq:case2Rz} and \eqref{eq:case2Rs}, we have
\begin{equation}
 	\frac{1}{n} I(\phi;f|\mathcal{C})<\delta(\epsilon) \label{eq:case2term1}
\end{equation}
for sufficiently large $n$.
Following the same arguments, we show that
\begin{equation}
 	\frac{1}{n}I(\psi;g|\mathcal{C})<\delta(\epsilon) \label{eq:case2term2}
\end{equation}
for sufficiently large $n$.

We then consider the term $I(\psi,g;Z^n|\mathcal{C})$ and have
\begin{flalign*}
	I(\psi,g;Z^n|\mathcal{C})&=I(\psi,g,X^n;Z^n|\mathcal{C})-I(X^n;Z^n|\psi,g,\mathcal{C})\\
	&=I(X^n;Z^n|\mathcal{C})-I(X^n;Z^n|\psi,g,\mathcal{C})\\
	&=H(X^n|\mathcal{C})-H(X^n|Z^n,\mathcal{C})-H(X^n|\psi,g,\mathcal{C})+H(X^n|Z^n,\psi,g,\mathcal{C})
\end{flalign*}
where	
\begin{flalign*}
	H(X^n|\psi,g,\mathcal{C})&=H(X^n,\psi,g|\mathcal{C})-H(\psi,g|\mathcal{C})\\
	&=H(X^n|\mathcal{C})-H(\psi,g|\mathcal{C})\\
	&\ge H(X^n|\mathcal{C})-n(R_\mathcal{X}+R_P).
\end{flalign*}
Hence,
\begin{equation*}
	I(\psi,g;Z^n|\mathcal{C})\le n(R_\mathcal{X}+R_P-H(X|Z))+H(X^n|Z^n,\psi,g,\mathcal{C})
\end{equation*}

Similarly to Lemma \ref{le:1}, we can show that if
\begin{equation}\label{eq:case2req3}
	R_\mathcal{X}+R_P<H(X|Z)-2\delta(\epsilon),
\end{equation}
then,
\begin{equation}
	\limsup_{n\rightarrow \infty}\frac{1}{n} H(X^n|Z^n,\psi,g,\mathcal{C})<H(X|Z)-R_\mathcal{X}-R_P+\delta(\epsilon).
\end{equation}
Consequently,
\begin{equation}
	\frac{1}{n} I(\psi,g;Z^n|\mathcal{C})<\delta(\epsilon) \label{eq:case2term3}
\end{equation}
for sufficiently large $n$.  This also implies that
\begin{equation}
	\frac{1}{n} I(g;Z^n|\mathcal{C})<\delta(\epsilon) \label{eq:case2term4}
\end{equation}
for sufficiently large $n$. Therefore, substituting \eqref{eq:case2term1}, \eqref{eq:case2term2}, \eqref{eq:case2term3} and \eqref{eq:case2term4} into \eqref{eq:case2ksbound} and \eqref{eq:case2kpbound}, we show that the leakage rates vanish for large enough $n$.

\underline{Uniformity}: Following from Lemma 22.2 in \cite{Gamal11}, we conclude that if $R_S<H(Z)-4\delta(\epsilon)$, then
\begin{equation}
\liminf_{n \rightarrow \infty}\frac{1}{n} H(K_S|\mathcal{C})\ge R_S-\delta(\epsilon),
\end{equation}
and if $R_P<H(X)-4\delta(\epsilon)$, then
\begin{equation}
\liminf_{n \rightarrow \infty}\frac{1}{n} H(K_P|\mathcal{C})\ge R_P-\delta(\epsilon),
\end{equation}
which prove the uniformity of the two keys.

\underline{Existence of a Codebook}:
We finally note that we have shown that
\begin{flalign*}
&\Pr\{K_S=\tilde{K}_S=\hat{K}_S\}+\Pr\{K_P=\hat{K}_P\}+I(K_S;\mathbf{F}|\mathcal{C})+I(K_P;\mathcal{F}Z^n|\mathcal{C})\\
&+\big[R_S-\frac{1}{n}H(K_S|\mathcal{C})\big]+\big[R_P-\frac{1}{n}H(K_P|\mathcal{C})\big]
\end{flalign*}
converges to zero as $n\rightarrow \infty$. This implies
\begin{flalign*}
&\mE_\mathcal{C}\Big\{ \Pr\{K_S=\tilde{K}_S=\hat{K}_S\big|\mathcal{C}=c\}+\Pr\{K_P=\hat{K}_P\big|\mathcal{C}=c\}+I(K_S;\mathbf{F}|\mathcal{C}=c)+I(K_P;\mathcal{F}Z^n|\mathcal{C}=c)\\
&+\big[R_S-\frac{1}{n}H(K_S|\mathcal{C}=c)\big]+\big[R_P-\frac{1}{n}H(K_P|\mathcal{C}=c)\big]\Big\}
\end{flalign*}
converges to zero as $n\rightarrow \infty$. Thus, there must exist one codebook $\mathcal{C}$ such that each term converges to zero as $n\rightarrow \infty$ due to non-negativity of all terms. Therefore, such a codebook satisfies all requirements simultaneously.

\subsection{Intuitive Justification of Secrecy}\label{ssec:understanding2}

%
%
 In this subsection, we intuitively explain that the generated $K_S$ and $K_P$ satisfy the secrecy requirements \eqref{eq:SKsecrecy} and \eqref{eq:PKsecrecy}.

We first justify that $K_S$, which is set as $\phi(Z^n)$, is almost independent from the public communication. Firstly, $\phi(Z^n)$, as the sub-bin index, is assigned independently from the bin index $f(Z^n)$, and hence is almost independent from the public transmission by $\mathcal{Z}$. It is then sufficient to justify that $\phi$ is almost independent of $g(X^n)$ given $f$. Given $g$, the bin $B_\mathcal{X}(g)$ contains $2^{nI(XZ;Y)}$ typical $x^n$ sequences on average. This implies that there are the same number $2^{n(I(XZ;Y)-I(X;Z))} \ge 1$ of $x^n$ that are jointly typical with any typical $z^n$ in $B_\mathcal{Z}(f)$. Hence, the bin index $g$ does not distinguish among $z^n$ within the bin $B_\mathcal{Z}(f)$, and hence does not distinguish among $\phi(z^n)$. On the other hand, if the alphabet of $g$ is too large such that $|B_\mathcal{X}(g)|<2^{nI(X;Z)}$, then there must exist some $z^n$ in bin $B_\mathcal{Z}(f)$ for which joint typical $x^n$ does not exist in the bin $B_\mathcal{X}(g)$. In this case, $g$ reveals some information about $z^n$, which can suggest exclusion of $\phi$ indices of those $z^n$ from being the key.


We then justify that $K_P$, which is set as $\psi(X^n)$, is almost independent from the public communication and $Z^n$. It is clear that $\psi(X^n)$ is independent from $g(X^n)$. It is then sufficient to justify that $\psi$ is almost independent from $Z^n$ given $g$. On average, any sub-bin within the bin $g$ contains $2^{nI(X;Z)}$ typical sequences $x^n$. This implies that there exists one $x^n$ (on average) in each sub-bin that is jointly typical with a typical $z^n$. Hence, knowing $Z^n$ does not distinguish among sub-bins of $x^n$. This justifies $\psi(X^n)$ is almost independent from $Z^n$.


\section{Achievability Proof for Case 3}\label{sec:proof3}
In this section, we provide the achievability proof for case 2 with subsection \ref{ssec:technique3} containing the technical proof and subsection \ref{ssec:understanding3} containing further intuitive justification.

\subsection{Technical Proof}\label{ssec:technique3}
In this subsection, we show that the outer bound given in Theorem \ref{th:CapReg} for case 3 is achievable.
In case 3,  $R_A=\min\{R_A, R_B, R_C\}$. In fact, the only possible case is $R_A<R_C\le R_B$ \cite{Ye05con}, which implies
\begin{flalign}
	&I(X;Y)>I(Z;XY),\label{eq:case3ass1}\\
	&I(X;Z)\le I(Y;XZ),\\
	&I(Y;Z)\le I(X;YZ).
\end{flalign}
The case of equality with $R_A=R_C$ reduces to case 2.

The outer bound in Theorem \ref{th:CapReg} for case 3 is plotted in Fig.~\ref{fig:case3} as the hexagon O-E-T-P-Q-A-O. It has been shown in \cite{Ye05con} that the corner points E, T and A are achievable.  The point P can be achieved by the same scheme as in case 2. It is thus sufficient to show that the point Q whose rate coordinates are given by $(I(Z;XY), I(X;Y)-I(Z;XY))$, is achievable. Then the entire hexagon can be achieved by time sharing. Our scheme to achieve the point Q is based on random binning and joint typicality.

\underline{Codebook Generation}:
At terminal $\mathcal{Z}$, randomly and independently assign a bin index $f$ to each sequence $z^n\in \mathcal{Z}^n$, where $f\in [1:2^{nR_\mathcal{Z}}]$ with $R_\mathcal{Z}$ given by
\begin{equation}\label{eq:case3Rz}
R_\mathcal{Z}=H(Z|X,Y)+\epsilon+2\delta(\epsilon).
\end{equation}
We use $f(z^n)$ to denote the bin index of the sequence $z^n$, and use $B_\mathcal{Z}(f)$ to denote the bin indexed by $f$. Then randomly and independently assign a sub-bin index $\phi$ to each sequence in each nonempty bin $B_\mathcal{Z}(f)$, where $\phi \in [1:2^{nR_S}]$
with $R_S$ given by
\begin{equation}\label{eq:case3Rs}
R_S=I(Z;XY)-2\epsilon-4\delta(\epsilon).
\end{equation}
We further use $B_\mathcal{Z}(f,\phi)$ to denote the sub-bin indexed by $\phi$ within the bin $B_\mathcal{Z}(f)$.

At terminal $\mathcal{X}$, randomly and independently assign a bin index $g$ to each sequence $x^n\in \mathcal{X}^n$, where $g\in[1:2^{nR_\mathcal{X}}]$ with $R_\mathcal{X}$ given by
\begin{equation} \label{eq:case3Rx}
	R_\mathcal{X}=H(X|Y)+\epsilon.
\end{equation}
We use $g(x^n)$ to denote the bin index of the sequence $x^n$, and use $B_\mathcal{X}(g)$ to denote the bin indexed by $g$. Then randomly and independently assign a sub-bin index $\psi$ to each sequence in each nonempty bin $B_\mathcal{X}(g)$, where $\psi \in [1:2^{nR_P}]$
with $R_P$ given by
\begin{equation}\label{eq:case3Rp}
R_P=I(X;Y)-I(Z;XY)-2\epsilon-2\delta(\epsilon).
\end{equation}
We further use $B_\mathcal{X}(g,\psi)$ to denote the sub-bin indexed by $\psi$ within the bin $B_\mathcal{X}(g)$.

At terminal $\mathcal{Y}$, randomly and independently assign a bin index $l$ to each sequence $y^n\in \mathcal{Y}^n$, where $l\in[1:2^{nR_\mathcal{Y}}]$ with $R_\mathcal{Y}$ given by
\begin{equation} \label{eq:case3Ry}
 	R_\mathcal{Y}=H(Y|X)-2\delta(\epsilon).
\end{equation}
We use $l(y^n)$ to denote the bin index of the sequence $y^n$, and use $B_\mathcal{Y}(l)$ to denote the bin indexed by $l$.

It can be verified that $R_P>0$ based on the case assumption \eqref{eq:case3ass1}.

This codebook assignment is revealed to all parties, i.e., terminals $\mathcal{X}, \mathcal{Y}, \mathcal{Z}$ and the eavesdropper.

\underline{Encoding and Transmission}:
Given a sequence $z^n$, terminal $\mathcal{Z}$ finds the index pair $(f,\phi)$ such that $z^n\in B_\mathcal{Z}(f,\phi)$, and then reveals the index $f=f(z^n)$ over the public channel to all parties, i.e., terminals $\mathcal{X},\mathcal{Y}$ and the eavesdropper.
Given a sequence $x^n$, terminal $\mathcal{X}$ finds the index pair $(g,\psi)$ such that $x^n\in B_\mathcal{X}(g,\psi)$, and then reveals the index $g=g(x^n)$ over the public channel to all parties, i.e., terminals $\mathcal{Y},\mathcal{Z}$ and the eavesdropper.
Given a sequence $y^n$, terminal $\mathcal{Y}$ finds the index $l$ such that $y^n \in B_\mathcal{Y}(l)$, and then reveals the index $l=l(y^n)$ over the public channel to all parties, i.e., terminals $\mathcal{X},\mathcal{Z}$ and the eavesdropper.

\underline{Decoding}:
Terminal $\mathcal{X}$, given $x^n$ and the bin indexes $f$ of $z^n$ and $l$ of $y^n$, claims $(\tilde{z}^n,\tilde{y}^n)$ as recovery of $(z^n,y^n)$ if there exists a unique pair $(\tilde{z}^n, \tilde{y}^n)$ such that $\tilde{z}^n\in B_\mathcal{Z}(f)$, $\tilde{y}^n\in B_\mathcal{Y}(l)$, and $(\tilde{z}^n, \tilde{y}^n, x^n) \in T^{(n)}_\epsilon(P_{XYZ})$, or claims decoding failure otherwise.

Terminal $\mathcal{Y}$, given $y^n$ and the bin indexes $f$ of $z^n$ and $g$ of $x^n$, claims $(\hat{z}^n,\hat{x}^n)$ as recovery of $(z^n,x^n)$ if there exists a unique pair $(\hat{z}^n, \hat{x}^n)$ such that $\hat{z}^n\in B_\mathcal{Z}(f)$, $\hat{x}^n\in B_\mathcal{X}(g)$, and $(\hat{z}^n, \hat{x}^n, y^n) \in T^{(n)}_\epsilon(P_{XYZ})$, or claims decoding failure otherwise.

We further assume that
\begin{flalign}
	H(Y|X)>H(Y|XZ). \label{eq:case3aspt}
\end{flalign}
The case of equality implies that the point Q coincides with the point P, and can hence be achieved using the scheme given in case 2.
Then due to \eqref{eq:case3Rz}, \eqref{eq:case3Rx} and \eqref{eq:case3Ry}, it can be verified that $R_\mathcal{Z}>H(Z|XY)$, $R_\mathcal{X}>H(X|YZ)$, $R_\mathcal{Y}>H(Y|XZ)$, $R_\mathcal{X}+R_\mathcal{Z}>H(XZ|Y)$ and $R_\mathcal{Y}+R_\mathcal{Z}>H(YZ|X)$ hold. It can then be shown that the following inequalities hold, according to the result of distributed source coding problem in \cite{Slepian73, Csiszar04, Gamal11}:
\begin{flalign}
	&\Pr\{Z^n\neq \tilde{Z}^n\ or\ Y^n \neq \tilde{Y}^n\}<\epsilon, \label{eq:case3errorx}\\
	&\Pr\{X^n \neq \hat{X}^n\ or\ Z^n \neq \hat{Z}^n\}<\epsilon. \label{eq:case3errory}
\end{flalign}

\underline{Key Generation}: Terminal $\mathcal{Z}$ claims $K_S=\phi(Z^n)$. Terminal $\mathcal{X}$ claims $\tilde{K}_S=\phi(\tilde{Z}^n)$ and $K_P=\psi(X^n)$. Terminal $\mathcal{Y}$ claims $\hat{K}_S=\phi(\hat{Z}^n)$ and $\hat{K}_P=\psi(\hat{X}^n)$. Due to \eqref{eq:case3errorx} and \eqref{eq:case3errory}, we have
\begin{flalign}
	&\Pr\{K_S=\tilde{K}_S=\hat{K}_S\}>1-\epsilon,\\
 	&\Pr\{K_P=\hat{K}_P\}>1-\epsilon.
\end{flalign}

\underline{Analysis of Secrecy}:
We evaluate  the key leakage rates averaged over the random codebook ensemble as follows. We let $l:=l(Y^n)$ and now $\mathbf{F}=\{f,g,l\}$. We then derive the following bounds:
\begin{flalign}
	I(K_S;\mathbf{F}|\mathcal{C})&=I(\phi;f,g,l|\mathcal{C}) \nn\\
 	&=I(\phi;f|\mathcal{C})+I(\phi;g,l|f,\mathcal{C}) \nn\\
 	&\le I(\phi;f|\mathcal{C})+I(\phi,f;g,l|\mathcal{C}) \nn\\
 	&\le I(\phi;f|\mathcal{C})+I(Z^n;g,l|\mathcal{C}) \label{eq:case3ksbound}
\end{flalign}
\begin{flalign}
 		I(K_P;\mathbf{F},Z^n|\mathcal{C})&=I(\psi;f,g,l,Z^n|\mathcal{C})\nn\\
 &=I(\psi;g,l,Z^n|\mathcal{C}) \nn\\
 		&= I(\psi;g|\mathcal{C})+I(\psi;l|g,\mathcal{C})+I(\psi;Z^n|g,l,\mathcal{C}) \nn\\
 		&\le I(\psi;g|\mathcal{C})+I(\psi,g;l|\mathcal{C})+I(\psi,g,l;Z^n|\mathcal{C}). \label{eq:case3kpbound}
\end{flalign}
We next show that each of the four terms $I(\phi;f|\mathcal{C}), I(\psi;g|\mathcal{C})$, $I(\psi,g;l|\mathcal{C})$ and $I(\psi,g,l;Z^n|\mathcal{C})$ can be arbitrarily small for large enough $n$. Following the same steps as in case 2 we can show that
\begin{flalign}
	&I(\phi;f|\mathcal{C})<\delta(\epsilon),\label{eq:case3term1}\\
	&I(\psi;g|\mathcal{C})<\delta(\epsilon) \label{eq:case3term2}
\end{flalign}
for large enough $n$. We then consider the term $I(\psi,g;l|\mathcal{C})$, and have
\begin{flalign*}
	I(\psi,g;l|\mathcal{C})&\le I(\psi,g,X^n;l|\mathcal{C}) \\&=I(X^n;l|\mathcal{C})\\
	&=I(X^n;l,Y^n|\mathcal{C})-I(X^n;Y^n|l,\mathcal{C})\\
	&=I(X^n;Y^n|\mathcal{C})-I(X^n;Y^n|l,\mathcal{C})\\
	&=H(Y^n|\mathcal{C})-H(Y^n|X^n,\mathcal{C})-H(Y^n|l,\mathcal{C})+H(Y^n|X^n,l,\mathcal{C})\\
	&\stackrel{(a)}{\le} H(Y^n|\mathcal{C})-H(Y^n|X^n,\mathcal{C})-(H(Y^n|\mathcal{C})-nR_\mathcal{Y})+H(Y^n|X^n,l,\mathcal{C})\\
	&\stackrel{(b)}{=} n(R_\mathcal{Y}-H(Y|X))+H(Y^n|X^n,l,\mathcal{C})
\end{flalign*}
where (a) follows because
\begin{equation*}
H(Y^n|l,\mathcal{C})=H(Y^n,l|\mathcal{C})-H(l|\mathcal{C})=H(Y^n|\mathcal{C})-H(l|\mathcal{C})\ge H(Y^n|\mathcal{C})-nR_\mathcal{Y},
\end{equation*}
and (b) follows because
\begin{equation*}
H(Y^n|X^n,\mathcal{C})=H(Y^n|X^n)=nH(Y|X).
\end{equation*}

Similarly to Lemma \ref{le:1}, we can show that if
\begin{equation}\label{eq:case3req1}
	R_\mathcal{Y}\le H(Y|X)-2\delta(\epsilon),
\end{equation}
then
\begin{equation}
	\limsup_{n \rightarrow \infty}\frac{1}{n}H(Y^n|X^n,l,\mathcal{C})<H(Y|X)-R_\mathcal{Y}+\delta(\epsilon).
\end{equation}
Consequently, we obtain
\begin{equation}
 	\frac{1}{n}I(\psi,g;l|\mathcal{C})<\delta(\epsilon) \label{eq:case3term3}
\end{equation}
for sufficiently large $n$.

For the term $I(\psi,g,l; Z^n|\mathcal{C})$, we derive the following bound:
\begin{flalign*}
		&I(\psi,g,l; Z^n|\mathcal{C})\\
		&=I(\psi,g,l,X^n,Y^n;Z^n|\mathcal{C})-I(X^n,Y^n;Z^n|\psi,g,l,\mathcal{C})\\
		&=I(X^n,Y^n;Z^n|\mathcal{C})-I(X^n,Y^n;Z^n|\psi,g,l,\mathcal{C})\\
		&=H(X^n,Y^n|\mathcal{C})-H(X^n,Y^n|Z^n,\mathcal{C})-H(X^n,Y^n|\psi,g,l,\mathcal{C})+H(X^n,Y^n|Z^n,\psi,g,l,\mathcal{C})\\
		&\stackrel{(a)}{\le}H(X^n,Y^n|\mathcal{C})-H(X^n,Y^n|Z^n,\mathcal{C})\\
		&-(H(X^n,Y^n|\mathcal{C})-n(R_\mathcal{X}+R_\mathcal{Y}+R_P))+H(X^n,Y^n|Z^n,\psi,g,l,\mathcal{C})\\
		&\stackrel{(b)}{=} n(-H(X,Y|Z)+R_\mathcal{X}+R_\mathcal{Y}+R_P)+H(X^n,Y^n|Z^n,\psi,g,l,\mathcal{C})
\end{flalign*}
where (a) follows because
\begin{flalign*}
	H(X^n,Y^n|\psi,g,l,\mathcal{C})&=H(X^n,Y^n,\psi,g,l|\mathcal{C})-H(\psi,g,l|\mathcal{C})\\
	&=H(X^n,Y^n|\mathcal{C})-H(\psi,g,l|\mathcal{C})\\
	&\ge H(X^n,Y^n|\mathcal{C})-n(R_\mathcal{X}+R_\mathcal{Y}+R_P),
\end{flalign*}
and (b) follows because
\begin{equation*}
H(X^n,Y^n|Z^n,\mathcal{C})=H(X^n,Y^n|Z^n)=nH(X,Y|Z).
\end{equation*}
Similarly to Lemma \ref{le:1}, we can show that if
\begin{equation}\label{eq:case3req2}
	R_\mathcal{X}+R_\mathcal{Y}+R_P<H(X,Y|Z)-2\delta(\epsilon),
\end{equation}
then,
\begin{equation}
	\limsup_{n \rightarrow \infty}\frac{1}{n}H(X^n,Y^n|Z^n,\psi,g,l,\mathcal{C})<H(X,Y|Z)-R_\mathcal{X}-R_\mathcal{Y}-R_P+\delta(\epsilon).
\end{equation}
Consequently, we have
\begin{equation}
	\frac{1}{n}I(\psi,g,l;Z^n|\mathcal{C})<\delta(\epsilon) \label{eq:case3term4}
\end{equation}
for large enough $n$. This also implies that
\begin{equation}
	\frac{1}{n} I(g,l;Z^n|\mathcal{C})<\delta(\epsilon) \label{eq:case3term5}
\end{equation}
for sufficiently large $n$. Therefore, substituting \eqref{eq:case3term1}, \eqref{eq:case3term2}, \eqref{eq:case3term3}, \eqref{eq:case3term4} and \eqref{eq:case3term5} into \eqref{eq:case3ksbound} and \eqref{eq:case3kpbound}, we show that the leakage rates vanish for large enough $n$.

\underline{Uniformity}: Following from Lemma 22.2 in \cite{Gamal11}, we conclude that if $R_S<H(Z)-4\delta(\epsilon)$, then
\begin{equation*}
	\liminf_{n \rightarrow \infty}\frac{1}{n} H(K_S|\mathcal{C})\ge R_S-\delta(\epsilon),
\end{equation*}
and if $R_P<H(X)-4\delta(\epsilon)$, then
\begin{equation*}
	\liminf_{n \rightarrow \infty}\frac{1}{n} H(K_P|\mathcal{C})\ge R_P-\delta(\epsilon),
\end{equation*}
which prove the uniformity of the two keys.

%

\underline{Existence of a Codebook}: This can be argued in the similar way as for case 2.

\subsection{Intuitive Justification of Secrecy}\label{ssec:understanding3}

%
In this subsection, we intuitively explain that the generated secret and private keys $K_S$ and $K_P$ satisfy the secrecy requirements \eqref{eq:SKsecrecy} and \eqref{eq:PKsecrecy}.

We first justify that $K_S$, which is set as $\phi(Z^n)$, is almost independent from the public communication. Firstly, it is clear that $\phi(Z^n)$  is almost independent from $f$. It is then sufficient to justify that $\phi(Z^n)$ is almost independent of $g(X^n)$ and $l(Y^n)$ given $f(Z^n)$.
For any given $g$ and $l$, there are on average $2^{nI(X;Y)}$ jointly typical pairs of $(x^n, y^n)$ such that $(x^n,y^n)\in B_\mathcal{X}(g)\times B_\mathcal{Y}(l)$. This implies that there are $2^{n(I(X;Y)-I(XY;Z))}\ge 1$ pairs of $(x^n,y^n)$ being jointly typical with any typical $z^n$ in $B_\mathcal{Z}(f)$. Hence, the bin indexes $g$ and $l$ do not distinguish among $z^n$ within the bin $B_\mathcal{Z}(f)$, and hence do not distinguish among the index $\phi$ of $z^n$.

We then justify that $K_P$, which is set as $\psi(X^n)$, is almost independent from the public communication and $Z^n$. It is clear that $\psi(X^n)$ is independent from $g(X^n)$. Similarly to case 2, we can argue that $\psi$ is almost independent from the bin index $l$ of $Y^n$ given $g$. It is then sufficient to justify that $\psi$ is independent from $Z^n$ given $l$ and $g$. On average, any sub-bin $B_\mathcal{X}(g,\psi)$ within the bin $B_\mathcal{X}(g)$ contains $2^{nI(XY;Z)}$ typical sequences $x^n$, and thus for a given $l$, there are $2^{nI(XY;Z)}$ jointly typical pairs of $(x^n,y^n)$ in each $B_\mathcal{X}(g,\psi)\times B_\mathcal{Y}(l)$ for any $\psi$. This implies that there exists one pair $(x^n, y^n)$ (on average) in each $B_\mathcal{X}(g,\psi)\times B_\mathcal{Y}(l)$ that is jointly typical with a typical $z^n$. Hence, given $l$ and $g$, knowing $Z^n$ does not distinguish among sub-bins of $x^n$. This justifies that, knowing $Z^n$ and the bin index pair $(g,l)$, one does not have preference of determining the sub-bin in which the true $X^n$ lies.


\section{Conclusion}\label{sec:conclusion}


In this paper, we have studied the three-terminal source-type model of simultaneously generating a secret and private key pair. We have shown that random binning and joint decoding schemes achieve an existing outer bound on the SK-PK capacity region established in \cite{Ye05con} for two cases. Hence, jointly with the capacity region established in \cite{Ye05con} for one case, the SK-PK capacity region for this model is characterized in general. As future work, we will extend this study to more general networks with more than three terminals. We will also apply the idea of our achievable schemes to other multi-key generation source models.

\vspace{10mm}

\appendix

\noindent {\Large \textbf{Appendix}}

\section{Proof of Lemma \ref{le:1}} \label{app:a}
The proof adapts the proof of Lemma 22.3 in \cite{Gamal11} with variations. For the sake of completeness, we provide the detail here.

Let
\begin{equation}
	E_1=\begin{cases}
		1, &\text{ if }Z^n\notin T^{(n)}_\epsilon(P_Z),\\
		0, &\text{ otherwise}.
	\end{cases}
\end{equation}
Hence $\Pr\{E_1=1\}\rightarrow 0$ as $n\rightarrow \infty$.

We have the following bound:
\begin{flalign}
	&H(Z^n|f,\phi, \mathcal{C}) \nn \\
	&\le H(Z^n, E_1|f,\phi,\mathcal{C}) \nn \\
	&=H(E_1|f,\phi,\mathcal{C})+H(Z^n|E_1,f,\phi,\mathcal{C}) \nn \\
	&\le 1+n\Pr\{E_1=1\}\log|\mathcal{Z}|+H(Z^n|E_1=0,f,\phi,\mathcal{C}) \nn \\
	&=1+n\Pr\{E_1=1\}\log|\mathcal{Z}|+\sum_{(i,j)} \Pr\left(f=i,\phi=j\big|E_1=0\right)H(Z^n|E_1=0,f=i,\phi=j,\mathcal{C}). \label{eq:aa1}
\end{flalign}

For a given codebook $\mathcal{C}$, let $N(\mathcal{C})$ be the number of sequences $z^n\in B_\mathcal{Z}(i,j)\cap T_\epsilon^{(n)}(P_Z)$. Define
\begin{equation}
	E_2(\mathcal{C})=\begin{cases}
		1, &\text{ if }N(\mathcal{C})\ge 2\mE[N(\mathcal{C})],\\
		0, &\text{ otherwise}.
	\end{cases}
\end{equation}

Note that $N(\mathcal{C})\sim Binomial(|T_\epsilon^{(n)}(P_Z)|, 2^{-n(R_S+R_\mathcal{Z})})$. Thus,
\begin{flalign}
	\mE[N(\mathcal{C})]=2^{-n(R_S+R_\mathcal{Z})}|T_\epsilon^{(n)}(P_Z)|,\\
	\Var[N(\mathcal{C})]\le 2^{-n(R_S+R_\mathcal{Z})}|T_\epsilon^{(n)}(P_Z)|.
\end{flalign}

By Chebyshev Inequality, we have
\begin{flalign}
	\Pr\{E_2(\mathcal{C})=1\}\le \frac{\Var[N(\mathcal{C})]}{(\mE[N(\mathcal{C})])^2}\le 2^{-n[H(Z)-R_S-R_\mathcal{Z}-\delta(\epsilon)]}.
\end{flalign}

Hence, if $R_S+R_\mathcal{Z}\le H(Z)-2\delta(\epsilon)$, then $\Pr\{E_2(\mathcal{C})=1\} \rightarrow 0 \ \ as \ \ n\rightarrow \infty$. Now,
\begin{flalign}
	&H(Z^n|E_1=0,f=i,\phi=j,\mathcal{C}) \nn \\
	&\quad\quad\le H(Z^n, E_2|E_1=0,f=i,\phi=j,\mathcal{C}) \nn \\
	&\quad\quad=H(E_2|E_1=0,f=i,\phi=j,\mathcal{C})+H(Z^n|E_2, E_1=0,f=i,\phi=j,\mathcal{C})  \nn \\
	&\quad\quad\le 1+n\Pr\{E_2=1\}\log|\mathcal{Z}|+H(Z^n|E_2=0, E_1=0,f=i,\phi=j,\mathcal{C}) \nn \\
	&\quad\quad\le 1+n\Pr\{E_2=1\}\log|\mathcal{Z}|+n(H(Z)-R_S-R_\mathcal{Z}+\delta(\epsilon)) \label{eq:aa2}
\end{flalign}

Substituting \eqref{eq:aa2} into \eqref{eq:aa1}, we conclude that if
\begin{equation}
	R_S+R_\mathcal{Z}\le H(Z)-2\delta(\epsilon),
\end{equation}
then
\begin{equation}
	\frac{1}{n}H(Z^n|f,\phi,\mathcal{C})\le H(Z)-R_S-R_\mathcal{Z}+\delta(\epsilon), \ \ \ as \ \ n\rightarrow \infty.
\end{equation}

\renewcommand{\baselinestretch}{1}
\begin{small}

\bibliographystyle{unsrt}
\bibliography{security}

\end{small}



%

%
%
%


\end{document}